\documentclass{cccg22}
\usepackage[T1]{fontenc}
\usepackage[utf8]{inputenc}

\usepackage{amsfonts,amsmath,amssymb}
\usepackage{cite,microtype,graphicx}
\newtheorem{observation}[theorem]{Observation}

\title{Reflections in an Octagonal Mirror Maze}
\author{David Eppstein\thanks{Department of Computer Science, University of California, Irvine. Research supported in part by NSF grant CCF-2212129.}}

\date{ }

\begin{document}
\maketitle  

\begin{abstract}
Suppose we are given an environment consisting of axis-parallel and diagonal line segments with integer endpoints, each of which may be reflective or non-reflective, with integer endpoints, and an initial position for a light ray passing through points of the integer grid. Then in time polynomial in the number of segments and in the number of bits needed to specify the coordinates of the input, we can determine the eventual fate of the reflected ray.
\end{abstract}

\section{Introduction}

There are many problems in graphics and visibility testing where it is of interest to determine the path that would be taken by a ray of light, through an environment that may contain mirrors. Figure~\ref{fig:8reflex} gives a simple example of a problem of this type. Even for very restricted environments such as the one depicted, where the starting point of the ray and all endpoints of mirrored segments have integer coordinates and where the mirrors are all either axis-aligned or at $45^\circ$ angles to the axes, the path of such a ray may be very complicated, taking a number of reflections that may depend on the geometry of the input and not merely on its combinatorial complexity. For instance, a ray that bounces diagonally between two parallel mirrors on opposite sides of a long thin rectangle will only exit the rectangle after a number of bounces proportional to the aspect ratio of the rectangle, even though such an environment consists of only two segments. Nevertheless, in that simple two-mirror example, the eventual path of the ray is easy to predict, without resorting to separately simulating each bounce. What about for environments of more than two mirrors? Is it still easy to ray-trace reflected rays in such environments?

\begin{figure}[ht]
\centering\includegraphics[width=\columnwidth]{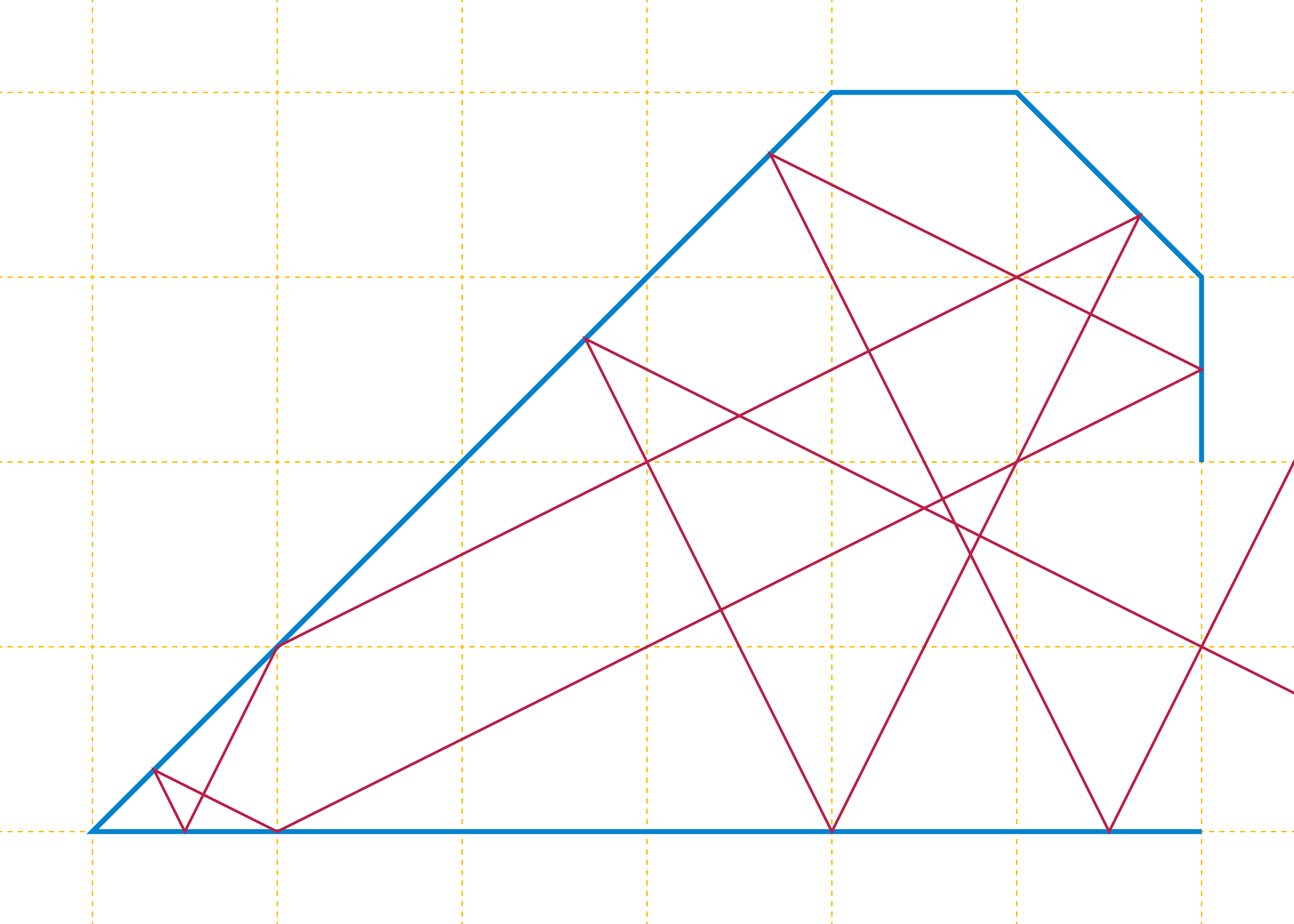}
\caption{The reflections of a light ray (red) among mirrors that are axis-aligned or at $45^\circ$ angles to the axes (blue)}
\label{fig:8reflex}
\end{figure}

We formalize this problem as follows. The input environment is described as a collection of line segments, with integer endpoints and either parallel to a coordinate axis or at a $45^\circ$ angle to the axes. Each side of each line segment may be marked as reflective or non-reflective. We are also given an integer position for the start of a light ray, and an integer vector describing the initial direction of the light ray. The restricted orientation of the mirrors ensures that each reflection of the ray in one of the reflective segments continues to have integer slope, on a line through infinitely many points of the integer grid. If, after repeated reflections, the ray eventually hits a non-reflective segment, the endpoint of a segment, or its initial position and orientation, it stops; otherwise, it must eventually escape the environment along an unobstructed infinite ray. The output of the problem is the eventual fate of the ray: the point where it stops, or the ray along which it escapes. Our main result is that we can determine this outcome in polynomial time.

Let $n$ denote the number of segments of the input, and suppose that all of the integers in the input specification (including the ones specifying the initial vector of the traced ray) have magnitude at most $N$. For these problem size parameters, it would be trivial to solve the problem in time polynomial in $N$ – simply trace the ray one reflection at a time – but this time bound is not polynomial, as it is exponentially larger than the input size. Our time bound is weakly polynomial, but not strongly polynomial: it is a polynomial of the number of bits required to specify the input, which is $O(n\log N)$. For the purposes of polynomial-time complexity, it would be equivalent to allow input coordinates that are rational numbers, rather than integers, with numerators and denominators of magnitude at most $N$. Clearing denominators would produce an integer input whose coordinates have magnitude $N^{O(n)}$ (the product of the input numerators and denominators), still requiring only a polynomial number of bits to specify, $O(n^2\log N)$.

The main idea of our algorithm is to transform the problem into one of determining the iterated behavior of a certain type of discrete one-dimensional dynamic system, which in a related recent paper~\cite{Epp-CIRC-21} we called an \emph{iterated integer interval exchange transformation}. In turn, following that paper, we can transform the iterated integer interval exchange transformation problem into a previously-studied problem in computational topology, of following paths along \emph{normal curves} on triangulated topological surfaces. To solve this path-following problem on normal curves in triangulated surfaces, we apply algorithms of Erickson and Nayyeri~\cite{EriNay-DCG-13}.

The general topic of visibility and ray-shooting with reflection is one with much prior work, both heuristic as part of the computer graphics rendering pipeline and with more rigorous bounds in computational geometry, for which see, e.g., \cite{AroDavDey-DCG-98a,AroDavDey-DCG-98b,PraPalDey-CGTA-98,OroPet-CCCG-01,AanBisPal-EuroCG-08,MahMohKou-CCCG-11,GhoGosMah-VC-12,VaeGho-TCS-19}. However, this past work has a combinatorial complexity that blows up with the number of reflections. In contrast, our results give a polynomial time algorithm whose complexity is independent of the number of reflections.

\section{Iterated interval exchange transformations}

In a recent paper of the author~\cite{Epp-CIRC-21} we investigated a broad class of problems, involving computing the $n$th iterate of a polynomial-time bijective function. One motivation for this investigation was in ray-tracing problems like the one studied here: if an environment consists only of mirrors, with no absorbing barriers for light, then (modulo representational issues involving whether reflections preserve the integer nature of a light ray) the mapping from each reflected position and direction of a light ray to the next is just such a polynomial-time bijection. Most of the problems considered in our recent paper have high computational complexity. However, our paper also identified a special class of bijections, the \emph{integer interval exchange transformations}, for which iterates can be computed in polynomial time. We will use the resulting \emph{iterated integer interval exchange transformation problem} as a subroutine in our algorithm for finding the result of a sequence of reflections. In this section we summarize the definitions needed to apply this problem, following our previous paper.

We define an integer interval exchange transformation to be a certain type of piecewise-linear bijective mapping on a range of consecutive integers. It may be defined by a partition of the range into subintervals, and by a permutation of those subintervals. The transformation then translates each subinterval (meaning that it acts on this interval by addition of the same value to each integer in the interval), so that the translated subintervals again form a partition of the same range, reordered into the given permutation.  An example, used in \cref{fig:normal-interval-exchange}, is the transformation on $[0,15)$ that permutes the intervals $a=[0,3],b=[4,5],c=[6],d=[7,14]$ into the permuted order $b,d,c,a$. This permutation describes the function
\[
x \mapsto \begin{cases}
x+11,&\text{for } x\in[0,3]\\
x-4,&\text{for } x\in[4,5]\\
x+4,&\text{for } x\in[6]\\
x-5,&\text{for } x\in[7,14].\\
\end{cases}
\]
A transformation of this type, with $m$ intervals on the range $[0,M-1]$, can be specified by $O(m\log M)$ bits of information, specifying the endpoints and permuted position of each subinterval. The resulting integer function may be evaluated on any integer $x$ in its range in time $O(m)$, by a sequential search of the listed subintervals to find the one containing $x$, and a second scan of the subintervals to determine which ones have permuted positions before the one containing $x$ and contribute to the translation offset for $x$. Even faster evaluations are possible if the intervals are stored in sorted order with their translation offsets. The \emph{iterated interval exchange transformation problem} takes as input an integer interval exchange transformation $\mu$, represented in either of these ways, an integer $x$ in its range, and another non-negative integer $k$. The goal is to compute $\mu^{(k)}(x)$, the result of repeatedly replacing $x$ by its transformed value, $k$ times.

Following a suggestion of Mark Bell, our paper~\cite{Epp-CIRC-21} shows that, for every integer interval exchange transformation, it is possible to find a corresponding triangulated two-dimensional manifold, and a \emph{normal curve} on the surface, such that tracing the normal curve for a specified number of steps corresponds to evaluating the integer interval exchange transformation  (Figure~\ref{fig:normal-interval-exchange}). Here, a normal curve is a curve through the triangles of the surface, avoiding triangle vertices and passing straight across each triangle from edge to edge, without crossing itself. It can be specified by a system of numbers on each edge counting the number of segments of the curve that cross that edge; this specification must obey certain consistency constraints (the numbers of crossings on the three edges of each triangle must obey the triangle inequality and sum to an even number). This specification is sufficient to reconstruct the curve itself, up to topological equivalence.

\begin{figure}[t]
\centering\includegraphics[width=\columnwidth]{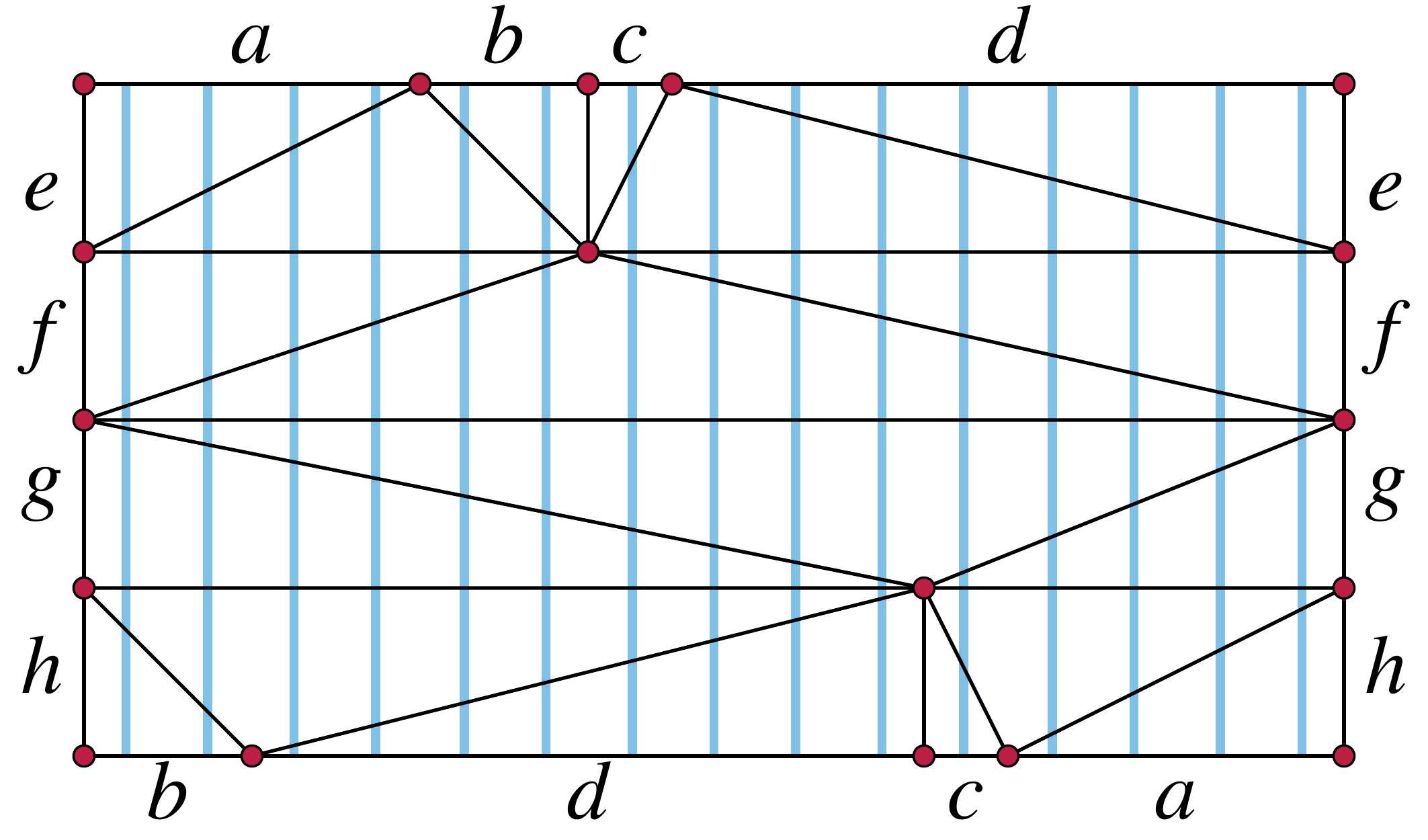}
\caption{A normal curve (light blue) on a triangulated double torus (black triangles and red vertices, glued from top to bottom and from left side to right side with the pairing indicated by the letters). Traversing the normal curve upwards from its central horizontal line, through the glued edges from top to bottom, and continuing upwards back to the same central line, permutes the branches of the curve according to the integer interval exchange transformation that maps $[0,3]\mapsto [11,14]$, $[4,5]\mapsto [0,1]$, $6\mapsto 10$, and $[7,14]\mapsto [2,9]$. From~\cite{Epp-CIRC-21}.}
\label{fig:normal-interval-exchange}
\end{figure}

In the transformation from our paper~\cite{Epp-CIRC-21}, the integers in the range $[0,M-1]$ of an integer interval exchange transformation are represented as the sequence of $M$ crossings of a normal curve, along a central horizontal edge of a triangulated surface. Each of these crossing points $x$ is connected by the normal curve, along a path of exactly $s$ segments for a number $s$ determined as part of the construction, to its image $\mu(x)$ according to the integer interval exchange transformation. Following this construction, the iterated interval exchange transformation problem can then be reduced to finding the crossing point that is $sk$ steps ahead of $x$ along the normal curve. This problem, of tracing paths for a given number of steps on a normal curve, has been given a polynomial-time solution by Erickson and Nayyeri~\cite{EriNay-DCG-13}. It follows that the iterated interval exchange transformation problem can also be solved in polynomial time. More precisely, the time is $O(m^2\log M)$, after an initial step in which the input parameter $k$ is reduced modulo the total number of steps in (a component of) the normal curve~\cite{Epp-CIRC-21}. As an integer division of a $\log k$-bit number by a $\log M$-bit number, this reduction step can be performed in time $O(\log k\log M)$ using naive division algorithms.

\section{Partial integer interval exchange}

Reflection in a mirror is a reversible transformation on systems of rays, but absorption by a non-reflective surface is not: many different rays could be absorbed at the same point. To mimic this non-reversibility in an integer exchange problem, while still allowing the polynomial-time procedure from our previous paper to apply, it is convenient to generalize the integer interval exchange problem to allow transformations that are only partially defined, as follows.

We define a \emph{partial integer interval exchange transformation}, for the range $[0,M-1]$, to be a system of disjoint subintervals of this range, together with a transformation that offsets each of these subintervals to another system of disjoint subintervals (necessarily of equal lengths). For instance, by omitting the subinterval $[6]$ from the previous example, we obtain a partial integer exchange transformation  that maps that maps $[0,3]\mapsto [11,14]$, $[4,5]\mapsto [0,1]$, and $[7,14]\mapsto [2,9]$. The \emph{domain} of the transformation is the union of the given subintervals, and the \emph{codomain} is the union of their target subintervals. This example has domain $[0,5]\cup[7,14]$ and codomain $[0,9]\cup[11,14]$.

\begin{lemma}
\label{lem:unco-undom}
If a partial integer interval exchange transformation is repeatedly applied to an input $x$ that does not belong to the codomain, it eventually reaches a transformed value that does not belong to the domain.
\end{lemma}

\begin{proof}
Consider the directed graph that connects each value to its transformed image. This graph has in-degree and out-degree at most one at each vertex, and has finitely many vertices, so it consists of a disjoint union of directed paths and directed cycles. An input $x$ that does not belong to the codomain has no incoming edge, so it is the starting vertex of a path, and is eventually transformed into the ending vertex of the same path, a value that does not belong to the domain.
\end{proof}

We define the \emph{iterated partial integer interval exchange transformation problem} to be a computational task that takes as input the description of a partial integer interval exchange transformation (as a system of subintervals and their targets) and a value $x$ that does not belong to the codomain, and that produces the corresponding value that does not belong to the domain, according to \cref{lem:unco-undom}.

\begin{lemma}
\label{lem:partial-alg}
The iterated partial integer interval exchange transformation problem can be solved in time $O(m^2\log M+\log^2 M)$, polynomial in the input representation size.\end{lemma}

\begin{proof}
We transform the problem in polynomial time into an equivalent instance of the (non-partial) iterated integer interval exchange transformation problem. Let $I_1, I_2, \dots$ be the intervals of the given partial transformation $f$, and let $f(I_1)$ etc. denote their images after the transformation. Suppose also that the given partial transformation operates on the range $[0,M-1]$ of length $M$. Let $m$ denote the total number of elements in this range that are missed by $f$: they are not in its domain. We define a new transformation $\bar f$ that operates on the range $[0,Mm+M+m-1]$ of length $Mm+M+m$, as follows:
\begin{itemize}
\item For each subinterval $I_i$ in the given transformation, with image $f(I_i)$, we include in $\bar f$ the mapping $I_i\mapsto f(I_i)$.
\item For each maximal subinterval $J_i$ of $[0,M-]\setminus\cup I_i$ (a subinterval not in the domain of $f$) we include in $\bar f$ a mapping from $J_i$ to a subinterval of $[M,M+m-1]$, so that the images of these subintervals are disjoint and completely cover $[M,M+m-1]$.
\item We include in $\bar f$ the mapping $[M,Mm+M-1]\mapsto [M+m,Mm+M+m-1]$. Iterating this mapping eventually transforms any value in $[M,M+m-1]$ to a value in $[Mm+M,Mm+M+m-1]$, but it takes $M$ iterations to do so.
\item For each maximal subinterval $K_i$ of $[0,M-]\setminus\cup f(I_i)$ (a subinterval not in the codomain of $f$) we include in $\bar f$ a mapping from a subinterval of $[Mm+M,Mm+M+m-1]$ to $K_i$, so that the preimages of these mappings are disjoint and completely cover $[Mm+M,Mm+M+m-1]$.
\end{itemize}
For instance, for the example partial integer interval exchange transformation $f$ given above, $M=15$ and $m=1$ (there is only one missing value from the transformation), and we have $\bar f$ mapping $[0,3]\mapsto [11,14]$, $[4,5]\mapsto [0,1]$, $[7,14]\mapsto [2,9]$; $[6]\mapsto [15]$; $[15,29]\mapsto [16,30]$; and $[30]\mapsto [10]$.

Suppose we apply the algorithm to the iterated integer interval exchange transformation problem, with transformation $\bar f$, on an input value $x$ that does not belong to the codomain, and that the output of this algorithm is a value $z$.
If we iterate $\bar f$ for a total of $M$ iterations, starting with a value $x$ that does not belong to the codomain, it will reach a value $y$ that does not belong to the domain in fewer than $M$ iterations, by \cref{lem:unco-undom}. The next iteration will map $y$ into a value $y'$ in the subinterval $[M,M+m-1]$, and the subsequent (again fewer than $M$) iterations will each add $m$ to this value $y'$. We may therefore obtain $y'$ by $z$ as the unique value in the subinterval $[M,M+m-1]$ that is congruent to $z$ modulo $m$. From $y'$, we may obtain the desired value $y$ as $\bar f^{-1}(y')$.

The time bound is obtained by plugging in the number of pieces of the resulting transformation, $O(m)$, the range of its values, $O(Mm)$, and the number of iterations, $O(M)$,  into the previous time bound for iterated integer interval exchange transformations.
\end{proof}

\section{Converting reflection to partial integer interval exchange}

The reason that we have restricted our attention to reflections in line segments that are horizontal, vertical, and diagonal is that these kinds of reflections preserve the integrality of the reflected rays. We formalize these observations as follows.

\begin{lemma}
\label{lem:8reflex}
If a ray whose direction is specified by a vector $(x,y)$ is reflected by a sequence of horizontal, vertical, or diagonal mirrors, then the resulting ray's direction can be specified by one of the eight vectors $(\pm x,\pm y)$ or $(\pm y,\pm x)$.
\end{lemma}

\begin{figure}[t]
\centering
\includegraphics[width=\columnwidth]{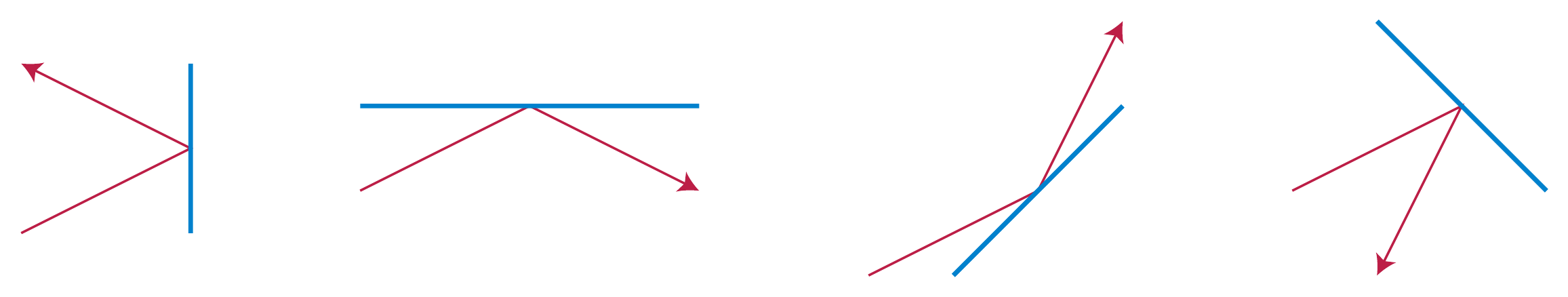}
\caption{Cases for reflection of a ray from a horizontal, vertical, or diagonal mirror}
\label{fig:reflex-cases}
\end{figure}

\begin{proof}
Vertical mirrors negate the first coordinate, horizontal mirrors negate the second coordinate, and diagonal mirrors exchange the two coordinates (possibly also negating both of them); see \cref{fig:reflex-cases}. The set of eight vectors of the lemma are preserved by these operations, so no other direction is possible.
\end{proof}

Along with these eight directions, it will also be important to keep track of the left-right orientation of each ray; that is, whether an image that follows a thickened copy of the ray has its orientation preserved or flipped. This changes on each reflection, and cannot be determined only from the ray's direction: a ray that reflects perpendicularly from a mirror will have its orientation changed, whereas a ray that returns from a corner-reflector (two perpendicular mirrors) will not.

\begin{lemma}
\label{lem:reflex-integer}
If a ray that passes through infinitely many points of the integer grid is reflected by a sequence of horizontal, vertical, or diagonal mirrors, each with integer endpoints, then the reflected ray again passes through infinitely many points of the integer grid.
\end{lemma}

\begin{proof}
When a ray is reflected by a mirror, it passes after the reflection through the same sequence of grid points that the unreflected ray would pass through in the reflection of the grid. But with mirrors of the type described by the lemma, the reflection of the grid is an integer grid with the same points.
\end{proof}

\begin{figure}[t]
\centering\includegraphics[width=\columnwidth]{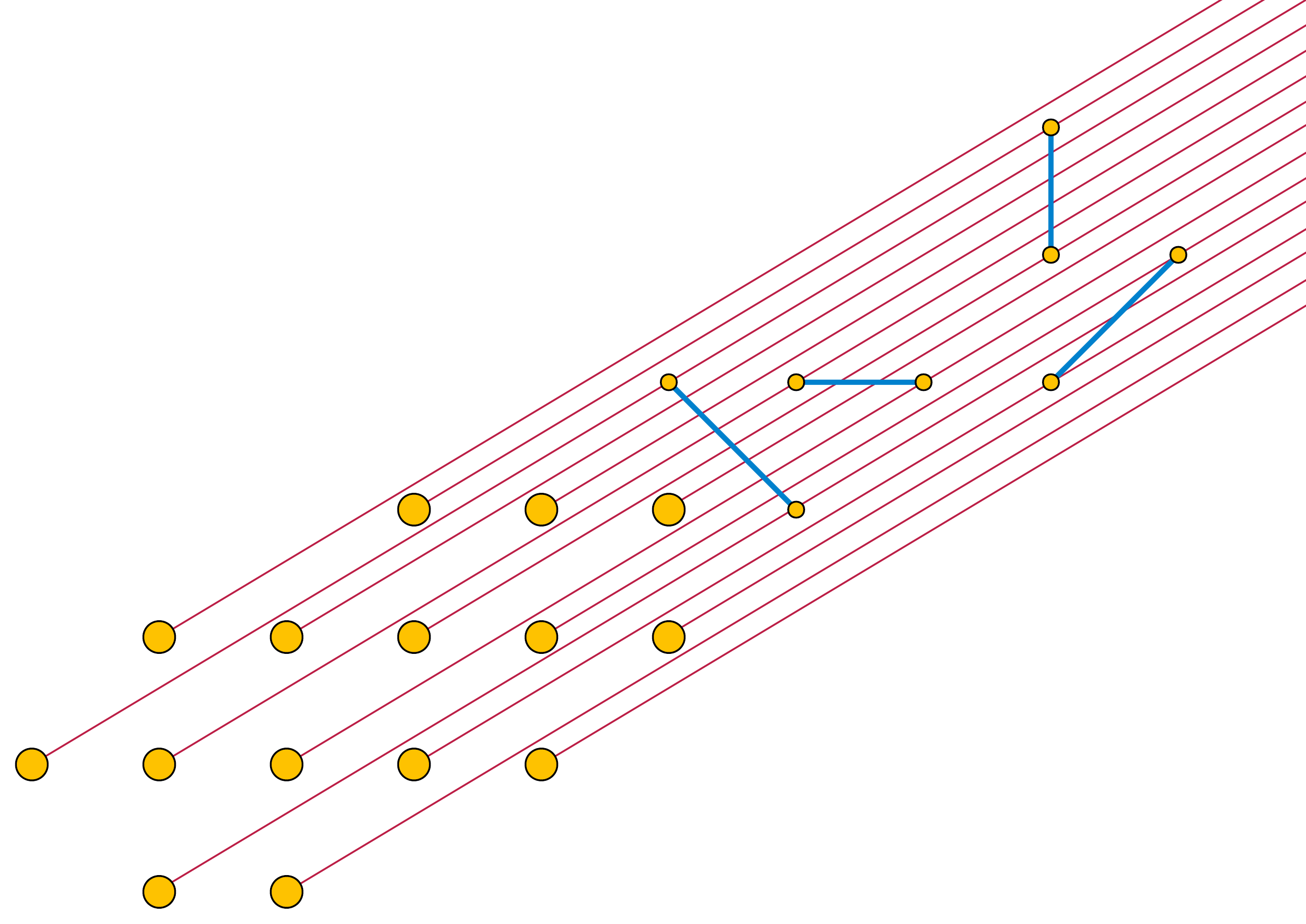}
\caption{Subdivision of grid segments and diagonals (blue) by the system of all lines through integer points in the direction of the vector $(5,3)$ (red).}
\label{fig:grid-subdiv}
\end{figure}

\begin{observation}
\label{obs:grid-subdiv}
Let $\mathcal{L}$ be the system of all lines in the plane that pass through integer points in the direction of a vector $(x,y)$, where $x$ and $y$ are integers in lowest terms (their greatest common divisor is one). Then the lines of $\mathcal{L}$ subdivide each vertical unit segment of  the integer grid into $x$ equal-length pieces, and each horizontal unit segment of the grid into $y$ equal-length pieces. They subdivide each diagonal segment of the grid, of length $\sqrt n$ with slope of the same sign as the slope of the lines in $\mathcal{L}$, into $|x-y|$ equal-length pieces, and each diagonal segment of the opposite slope into $|x+y|$ equal -length pieces. (See \cref{fig:grid-subdiv}.)
\end{observation}

Putting these observations together, we can transform any octagonal mirror maze into an equivalent partial integer interval exchange transformation.

\begin{lemma}
\label{lem:equiv-partial}
Suppose we are given an environment, described as a collection of line segments, each side of which may be marked as reflective or non-reflective, with integer endpoints, an integer position for the start of a light ray, and an integer vector describing the initial direction of the light ray. Then in polynomial time we can construct a partial integer interval exchange transformation $f$ whose values correspond to points of the environment (either the starting point or points along the segments of the environment) and directions of a reflected light ray emanating from that point, such that:
\begin{itemize}
\item If a value $v$ belongs to the domain of $f$, then the ray described by $v$ is eventually reflected by the environment, with its first reflection at a position and direction described by $f(v)$.
\item If a value $v$ does not belong to the domain of $f$, then the ray described by $v$ hits an absorbing barrier before being reflected, or escapes to infinity.
\end{itemize}
\end{lemma}

\begin{proof}
We surround the given environment with a non-reflective bounding box that will catch all escaping rays. By \cref{lem:8reflex} and \cref{lem:reflex-integer}, we need only consider rays through integer points, in eight directions and two left-right orientations. By \cref{obs:grid-subdiv}, we need only consider a discrete evenly-spaced set of possible reflection points along each segment of the environment, of a size that can be described by an integer with polynomially many bits. (Recall that we are not restricting the lengths of our grid segments or our initial direction to have polynomial magnitude, only to be integers with a polynomial number of bits.) We will create a partial integer interval exchange transformation $f$ whose range is partitioned into subintervals, one for each combination of an environment segment, a direction of the emanating ray, and a left-right orientation, with the length of these subintervals obtained by multiplying the length of the segment by the number of reflection points per unit length given by \cref{obs:grid-subdiv}. For each of these subintervals, we further partition it into sub-subintervals, in polynomial time, according to the next object in the environment that a ray in that direction would hit, as depicted in \cref{fig:reflex-subintervals}. (This is simply a lower envelope of $n$ disjoint line segments, for a projection direction determined by the ray direction.) When the next object to be hit is reflecting, we map a sub-subintervals to the interval describing the corresponding reflection points along that object. When, instead, it is absorbing, we omit that sub-subinterval from the partial integer interval exchange transformation.
\end{proof}

\begin{figure}[t]
\includegraphics[width=\columnwidth]{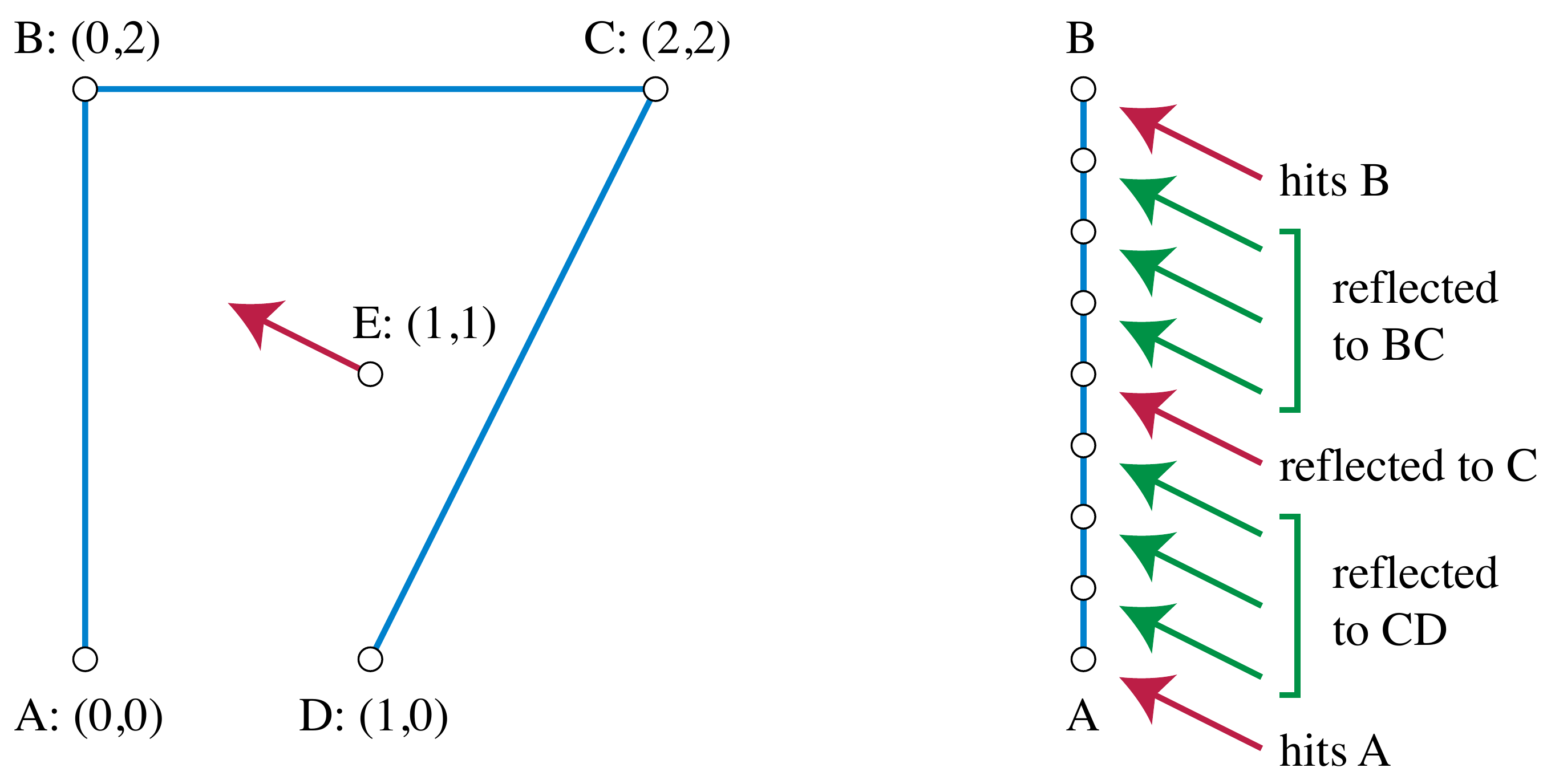}
\caption{Reflection points along segment $AB$ of a given environment, partitioned into subintervals according to
the next object in the reflected path for incoming rays of slope $-\tfrac12$.}
\label{fig:reflex-subintervals}
\end{figure}

\section{The main result}

Putting these components together, we have the following result:

\begin{theorem}
Suppose we are given an environment, described as a collection of line segments, each side of which may be marked as reflective or non-reflective, with integer endpoints, an integer position for the start of a light ray, and an integer vector describing the initial direction of the light ray. Then in time polynomial in the number of segments and in the number of bits needed to specify the integers of the input, we can determine whether the reflected ray is eventually absorbed or escapes to infinity. If it is absorbed, we can determine where it is absorbed, what direction it comes from when it is absorbed, and how many bounces it makes before this happens. If it escapes to infinity, we can determine its eventual escape path, and how many bounces it takes before reaching this path. The time bound for these algorithms is $O(n^2\log N+\log^2 N)$.
\end{theorem}

\begin{proof}
We convert the input to an equivalent partial integer interval exchange transformation according to \cref{lem:equiv-partial}, and then apply the polynomial-time algorithm for the iterated partial integer interval exchange transformation problem of \cref{lem:partial-alg}.

An input of size $n$ and coordinate magnitude $N$ can be converted into a partial integer interval exchange transformation whose number of subintervals is $O(n)$ (each comes from a trapezoid in four trapezoidal decompositions with sides parallel to one of the directions of the reflected rays) and whose range is $O(N^2)$ (combining the magnitude of the environment coordinates with the number of reflection points along each grid segment). For inputs of this size, the time to apply an algorithm for the iterated partial integer interval exchange transformation is $O(n^2\log N+\log^2 N)$.
\end{proof}

It would be of interest to determine to what extent this algorithm can be generalized. Can we determine the geometric length of the reflected path of a light ray, and not just its number of bounces? Are there other systems of slopes, beyond the axis-aligned and diagonal slopes, for which similar algorithms can work? How does the complexity of the algorithm depend on the system of slopes? For slopes that do not preserve the rationality of reflected rays, what can be said about the computational complexity of the problem?

\bibliographystyle{plainurl}
\bibliography{octagonal}

\begin{thebibliography}{10}

\bibitem{AanBisPal-EuroCG-08}
Mridul Aanjaneya, Arijit Bishnu, and Sudebkumar~Prasant Pal.
\newblock {Directly visible pairs and illumination by reflections in orthogonal
  polygons}.
\newblock In Sylvain Petitjean, editor, {\em Collection of abstracts of the
  24th European Workshop on Computational Geometry}, pages 241{--}244.
  INRIA-LORIA, March 18--20 2008.
\newblock URL:
  \url{https://physbam.stanford.edu/~aanjneya/mridul_files/papers/vis.pdf}.

\bibitem{AroDavDey-DCG-98b}
Boris Aronov, Alan~R. Davis, Tamal~K. Dey, Sudebkumar~P. Pal, and D.~Chithra
  Prasad.
\newblock {Visibility with multiple reflections}.
\newblock {\em Discrete {\&} Computational Geometry}, 20(1):61{--}78, 1998.
\newblock \href {http://dx.doi.org/10.1007/PL00009378}
  {\path{doi:10.1007/PL00009378}}.

\bibitem{AroDavDey-DCG-98a}
Boris Aronov, Alan~R. Davis, Tamal~K. Dey, Sudebkumar~P. Pal, and D.~Chithra
  Prasad.
\newblock {Visibility with one reflection}.
\newblock {\em Discrete {\&} Computational Geometry}, 19(4):553{--}574, 1998.
\newblock \href {http://dx.doi.org/10.1007/PL00009368}
  {\path{doi:10.1007/PL00009368}}.

\bibitem{Epp-CIRC-21}
David Eppstein.
\newblock {The complexity of iterated reversible computation}.
\newblock Electronic preprint arxiv:2112.11607, 2021.

\bibitem{EriNay-DCG-13}
Jeff Erickson and Amir Nayyeri.
\newblock {Tracing compressed curves in triangulated surfaces}.
\newblock {\em Discrete {\&} Computational Geometry}, 49(4):823{--}863, 2013.
\newblock \href {http://dx.doi.org/10.1007/s00454-013-9515-z}
  {\path{doi:10.1007/s00454-013-9515-z}}.

\bibitem{GhoGosMah-VC-12}
Subir~Kumar Ghosh, Partha~P. Goswami, Anil Maheshwari, Subhas~C. Nandy,
  Sudebkumar~Prasant Pal, and Swami Sarvattomananda.
\newblock {Algorithms for computing diffuse reflection paths in polygons}.
\newblock {\em The Visual Computer}, 28(12):1229{--}1237, 2012.
\newblock \href {http://dx.doi.org/10.1007/s00371-011-0670-z}
  {\path{doi:10.1007/s00371-011-0670-z}}.

\bibitem{MahMohKou-CCCG-11}
Salma~Sadat Mahdavi, Ali Mohades, and Bahram Kouhestani.
\newblock {Computing $k$-link visibility polygons in environments with a
  reflective edge}.
\newblock In {\em Proceedings of the 23rd Annual Canadian Conference on
  Computational Geometry, Toronto, Ontario, Canada, August 10-12, 2011}, 2011.
\newblock URL: \url{https://www.cccg.ca/proceedings/2011/papers/paper63.pdf}.

\bibitem{OroPet-CCCG-01}
Joseph O'Rourke and Octavia Petrovici.
\newblock {Narrowing light rays with mirrors}.
\newblock In {\em Proceedings of the 13th Canadian Conference on Computational
  Geometry, University of Waterloo, Ontario, Canada, August 13-15, 2001}, pages
  137{--}140, 2001.
\newblock URL: \url{https://www.cccg.ca/proceedings/2001/orourke-13443.ps.gz}.

\bibitem{PraPalDey-CGTA-98}
D.~Chithra Prasad, Sudebkumar~P. Pal, and Tamal~K. Dey.
\newblock {Visibility with multiple diffuse reflections}.
\newblock {\em Computational Geometry: Theory and Applications},
  10(3):187{--}196, 1998.
\newblock \href {http://dx.doi.org/10.1016/S0925-7721(97)00021-7}
  {\path{doi:10.1016/S0925-7721(97)00021-7}}.

\bibitem{VaeGho-TCS-19}
Arash Vaezi and Mohammad Ghodsi.
\newblock {Visibility extension via mirror-edges to cover invisible segments}.
\newblock {\em Theoretical Computer Science}, 789:22{--}33, 2019.
\newblock \href {http://dx.doi.org/10.1016/j.tcs.2019.02.011}
  {\path{doi:10.1016/j.tcs.2019.02.011}}.

\end{thebibliography}
\end{document}